\documentclass[12pt,notitlepage]{article}%
\usepackage{amsthm}
\usepackage{setspace}
\usepackage{eurosym}
\usepackage{subfigure}
\usepackage{xcolor}
\usepackage{amssymb}
\usepackage{mathrsfs}
\usepackage{graphicx}
\usepackage[colorlinks,linkcolor=links,citecolor=cites,urlcolor=MyDarkBlue]%
{hyperref}
\usepackage{amsfonts}
\usepackage{amsmath}
\usepackage{amstext}
\usepackage{appendix}
\usepackage{appendix}
\usepackage{mathpazo}
\usepackage{tikz}
\usepackage{verbatim}

\usetikzlibrary{trees}
\usetikzlibrary{matrix}
\usetikzlibrary{decorations.pathreplacing}
\usetikzlibrary{topaths,calc}
\usepackage{natbib}%
\providecommand{\U}[1]{\protect\rule{.1in}{.1in}}
\newtheorem{theorem}{Theorem}

\newtheorem{definition}{Definition}
\newtheorem{example}{Example}

\numberwithin{equation}{section}

\definecolor{MyDarkBlue}{rgb}{0,0.08,0.45}
\definecolor{cites}{HTML}{324b13}
\definecolor{links}{HTML}{1a663b}
\definecolor{MyLightMagenta}{cmyk}{0.1,0.8,0,0.1}
\hypersetup{
colorlinks,citecolor=blue,filecolor=black,linkcolor=blue,urlcolor=blue
}
\usepackage[top=0.9in,bottom=0.9in,left=0.9in,right=0.9in]{geometry}

\usepackage{setspace}
\usepackage{footmisc}
\usepackage{lipsum}

\begin{document}

\title{Concave many-to-one matching}
\author{Chao Huang\thanks{Institute for Social and Economic Research, Nanjing Audit University. Email: huangchao916@163.com.}}
\date{}
\maketitle

\begin{abstract}
We propose a notion of concavity in two-sided many-to-one matching, which is an analogue to the balancedness condition in cooperative games. A stable matching exists when the market is concave. We provide a class of concave markets. In the proof of the existence theorem, we use Scarf's algorithm to find a stable schedule matching, which is of independent interest.
\end{abstract}



\section{Introduction}\label{Sec_intro}

This paper studies the existence of a stable matching in many-to-one matching with nontransferable utilities. We propose a concavity condition using schedule matchings. In a schedule matching, each worker schedules her time among her acceptable firms, and each firm schedules its time among its acceptable groups of workers. To distinguish from schedule matchings, we sometimes call a matching a full-time matching. We show that a stable full-time matching exists if the market satisfies the condition
\begin{equation}\label{balanced}
\begin{aligned}
&\qquad\text{for any schedule matching, \textbf{if we consider each worker}}\\
&\text{\textbf{to be slightly better off when her employer becomes better}}\\
&\text{\textbf{off},\, there is a full-time matching\, in which\, each agent is not}\\
&\text{worse off than in its worst situation in the schedule matching.}
\end{aligned}
\end{equation}
We call a market satisfying this condition concave.\footnote{This condition may be called balancedness as it is the counterpart to the balancedness condition of \cite{S67}. We use the current name because (i) we want to avoid confusion with the balanced firm-worker hypergraph studied by \cite{H22}; (ii) the term ``concave'' seems to be more informative than ``balanced'' in summarizing condition (\ref{balanced}); (iii) if each agent has a transferable utility, the Bondareva-Shapley balancedness condition (see \citealp{M98}) is equivalent to the concavity of the aggregate valuation (see, for example, Lemma 2 of \citealp{TY19}), and both conditions characterize the existence of a stable matching. Our condition can also be viewed as the counterpart to the concavity of the aggregate valuation in matching with transferable utilities.} The artificial externalities stated in the sentence in bold are proposed by \cite{K10}. Condition (\ref{balanced}) without the sentence in bold is an analogue to the balancedness condition of \cite{S67} for a nonempty core. We use an example to illustrate the above statements. Consider a market with two firms $f_1,f_2$, two workers $w_1,w_2$, and the following preferences.
\begin{equation}\label{exam_in1}
\begin{aligned}
&f_1: \{w_1,w_2\}\succ\{w_2\}\succ\emptyset \qquad\qquad\qquad\qquad &w_1: &\quad f_1\succ f_2\\
&f_2: \{w_1\}\succ\{w_2\}\succ\emptyset \qquad\qquad\qquad\qquad &w_2: &\quad f_2\succ f_1
\end{aligned}
\end{equation}
Recall that an individually rational matching is stable if no firm-set of workers coalition can block the matching such that the firm becomes better off and no worker is worse off. The market (\ref{exam_in1}), which is modified from an example of \cite{CKK19}, has no stable (full-time) matching. If firm $f_2$ hires worker $w_1$, and firm $f_1$ hires worker $w_2$, then worker $w_1$ will go to firm $f_1$; if firm $f_2$ hires worker $w_2$, and firm $f_1$ is unmatched, then firm $f_2$ will fire worker $w_2$ and hire worker $w_1$; if firm $f_2$ is unmatched, and firm $f_1$ hires both workers, then worker $w_2$ will go to firm $f_2$. By contrast, an outcome is in the core if no coalition can block the outcome such that each member of the coalition becomes better off. The matching $(f_1,w_2)(f_2,w_1)$ is in the core of this market. Notice that the coalition $\{f_1,w_1,w_2\}$ does not core-wise block this matching since worker $w_2$ does not become better off. In many-to-one matching, the set of stable matchings is a stronger notion than the core. The market (\ref{exam_in1}) has a nonempty core but has no stable matching since this market satisfies condition (\ref{balanced}) without the sentence in bold but fails the condition in complete form. For example, consider the full-time matching $(f_1,w_2), (f_2,w_1)$ and the schedule matching
\begin{equation}\label{intro_schedule}
\left(
             \begin{aligned}
             f_1\qquad\quad\\
             \frac{1}{2}\{w_1,w_2\}+\frac{1}{2}\emptyset
             \end{aligned}
\right)
\left(
             \begin{aligned}
             f_2\qquad\quad\\
             \frac{1}{2}\{w_1\}+\frac{1}{2}\{w_2\}
             \end{aligned}
\right)
\end{equation}
The schedule matching means that firm $f_1$ schedules its time over $\{w_1,w_2\}$ and $\emptyset$ equally, and firm $f_2$ schedules its time over $\{w_1\}$ and $\{w_2\}$ equally. Each agent is not worse off in the full-time matching than its worst situation in the schedule matching. For example, worker $w_2$'s worst situation in the schedule matching is working for firm $f_1$ with worker $w_1$, which is the same as her situation in the full-time matching: working for firm $f_1$ alone. But if we consider worker $w_2$ to be slightly better off when her employer becomes better off, since firm $f_1$ prefers $\{w_1,w_2\}$ to $\{w_2\}$, worker $w_2$ is worse off in the full-time matching than her worst situation in the schedule matching. The word ``slightly'' refers to that the increment in a worker's ``utility'' as her employer becomes better off is small such that the worker's preference over different employers is unchanged.

The existence of a stable matching under condition (\ref{balanced}) follows from  Scarf's lemma and an observation of \cite{K10}. The latter work showed that if we impose artificial externalities into workers' preferences such that each worker is slightly better off when her employer becomes better off, the set of stable matchings in the original market coincides with the core of the modified market. For example, in the market modified from (\ref{exam_in1}), the matching $(f_1,w_2)(f_2,w_1)$ is no longer in the core since $\{f_1,w_1,w_2\}$ is a blocking coalition in which each member becomes better off. Scarf's lemma (Theorem 2 of \citealp{S67}) implies that the modified market has a nonempty core under condition (\ref{balanced}) without the sentence in bold, and therefore a stable matching exists in the original market under condition (\ref{balanced}).

We can strengthen condition (\ref{balanced}) by introducing a notion of generalized schedule matching called $\pi$-schedule matching, which is an analogue to the $\pi$-balanced collection of coalitions proposed by \cite{B70}. A $\pi$-schedule matching generalizes a schedule matching by allowing different workers to have different labor supplies, and different firms to have different resource capcities. A worker can have different labor intensities when working with different employers and colleagues, and a firm can have different resource input intensities when hiring different sets of employees. A schedule matching is a $\pi$-schedule matching when the $\pi$-scheme specifies the same above parameters. We can generalize condition (\ref{balanced}) by replacing ``any schedule matching'' with ``any $\pi$-schedule matching for a $\pi$-scheme''. Therefore, a market may fails condition (\ref{balanced}) for this $\pi$-scheme but may satisfies the generalized condition for another $\pi$-scheme.

The market with a balanced firm-worker hypergraph studied by \cite{H22} is an instance of concave markets. In this paper, we provide another concave market. We assume workers are divided into a set of leaders and a set of followers. Each follower follows a leader. A set of workers is acceptable to a firm only if the set contains exactly one leader and includes only followers who follow this leader. We show that this market is concave, and we can find a stable matching by a variant Deferred Acceptance (henceforth, DA) algorithm.

We can find a stable matching in a general concave market using Scarf's algorithm. If a market is concave under a specific $\pi$ scheme, we can use Scarf's algorithm to produce a stable $\pi$-schedule matching with respect to this $\pi$ scheme, then the full-time matching for condition (\ref{balanced}) with respect to this stable $\pi$-schedule matching is stable. The computation of a stable $\pi$-schedule matching is of independent interest.

\subsection{Related literature\label{Sec_Lit}}

Studies on two-sided matching originated with the seminal work of \cite{GS62}. When there is no complementarity in agents' preferences, \cite{KC82}, \cite{RS90}, and \cite{HM05} showed that stable matchings exist in matching with transferable utilities, discrete matching, and matching with contracts, respectively. On the other hand, the balancedness condition for games with transferable utilities is obtained by \cite{B63} and \cite{Sh67} independently. \cite{S67} derived the balancedness condition for games with nontransferable utilities from an exchange economy with convex preferences. \cite{M98} showed that the Bondareva-Shapley balancedness condition characterizes the existence of an equilibrium in an exchange economy with indivisible goods and transferable utilities.

Our work is inspired by \cite{NV18,NV19}, who applied Scarf's lemma to discrete matching with responsive preferences in the presence of couples or proportional constraints. They showed that firms' quotas or proportional constraints can be modified slightly to obtain an integral vertex of the polytope in Scarf's lemma that corresponds to a stable matching. On the other hand, \cite{H23} proposed an integer feasibility program that corresponds to a transformation from a stable schedule matching into a stable matching. We find the feasibility program of the latter work similar as the polytope in Scarf's lemma. Nguyen and Vohra's works and this similarity motivate us to apply Scarf's method to a general many-to-one market. We cannot apply Scarf's lemma directly since the set of stable matchings is not equivalent to the core. The reduction of \cite{K10} is useful as it transforms the set of stable matchings of the original market into the core of a modified market. 

Schedule matching has been studied by \cite{BB02} in one-to-one matching and by \cite{AG03} in many-to-one matching with substitutable preferences. Schedule matching can also be viewed as matching probabilistically, the latter has been studied by \cite{KU15} in school choice. The existence theorem of \cite{CKK19} implies the existence of a stable schedule matching in a general many-to-one market. We can also derive condition (\ref{balanced}) from this result. Scarf's lemma provides an algorithm for computing a stable schedule matching and facilitates the extension of condition (\ref{balanced}) to the $\pi$-concavity.

When there is no complementarity in firms' preferences, \cite{E12} showed that firms and workers essentially bargain over one dimension despite that contracts may be multidimensional. See \cite{K12} and \cite{S15} for more general results. Concave markets are compatible with complementarities and allow multidimensional negotiations between firms and workers.

The remainder of this paper is organized as follows. Section \ref{Sec_M} presents the model of many-to-one matching and the problem of schedule matching. Section \ref{Sec_pi} introduces concave markets and our existence theorem. Section \ref{Sec_App} provides a class of concave markets. We introduce stability for schedule matching in Section \ref{Sec_stableSch} and use Scarf's algorithm to find a stable schedule matching in Section \ref{Sec_Alg}.

\section{Model\label{Sec_M}}

\subsection{Many-to-one matching}

There is a finite set $F$ of firms and a finite set $W$ of workers with $|F|,|W|\geq2$.\footnote{A stable matching always exists when there is only one firm (or one worker): Among all individually rational matchings, the one that the firm (resp. the worker) prefers most is stable.} Let $N\equiv F\cup W$ and $n\equiv|N|$. Let $X$ be a finite set of contracts in which each contract $x\in X$ is signed by one firm $x_F\in F$ and one worker $x_W\in W$. For each subset of contracts $Y\subseteq X$ and each agent $i\in N$, let $Y_i\equiv\{x\in Y|x_F=i \text{ or } x_W=i\}$ be the subset of $Y$ in which each contract is signed by agent $i$. For each subset of contracts $Y\subseteq X$, let $W(Y)\equiv\{x_W|x\in Y\}$ be the set of workers who have contracts in $Y$, $F(Y)\equiv\{x_F|x\in Y\}$ the set of firms who have contracts in $Y$, and $N(Y)\equiv W(Y)\cup F(Y)$ the set of agents who have contracts in $Y$. For each $w\in W$, define
\begin{equation*}
\mathcal{A}^w\equiv\{\{x\}|x\in X_w\}\cup\{\emptyset\}.
\end{equation*}
Each set from $\mathcal{A}^w$ is called an \textbf{assignment of worker $w$} in the market. Let $\mathcal{A}^W\equiv\cup_{w\in W}\mathcal{A}^w$ be the collection of workers' assignments. Each worker $w\in W$ has a complete, transitive, and strict preference $\succ_w$ over its possible assignments from $\mathcal{A}^w$. Notice that we define a worker's preference to be over singletons\footnote{A singleton is a set that contains exactly one element.} rather than contracts. We assume each contract $x\in X$ is acceptable for the associated worker $x_W$, that is, $\{x\}\succ_w\emptyset$ for each $w\in W$ and $x\in X_w$. If a worker considers a contract unacceptable, this contract is not included in $X$. For each worker $w\in W$ and any two of her assignments $Y,Z\in \mathcal{A}^w$, we write $Y\succeq_w Z$ if $Y\succ_w Z$ or $Y=Z$. For each firm $f\in F$, define
\begin{equation*}
\mathcal{A}^f\equiv\{Z\subseteq X_f||Z_w|\leq1\text{ for each } w\in W\}
\end{equation*}
to be the collection of contracts signed by firm $f$ in which there is at most one contract for each worker. Each set of contracts from $\mathcal{A}^f$ is called an \textbf{assignment of firm $f$} in the market. Let $\mathcal{A}^F\equiv\cup_{f\in F}\mathcal{A}^f$ be the collection of firms' assignments. Each firm $f\in F$ has a complete, transitive, and strict preference $\succ_f$ over its possible assignments from $\mathcal{A}^f$. For each firm $f\in F$ and any two of its assignments $Y,Z\in \mathcal{A}^f$, we write $Y\succeq_fZ$ if $Y\succ_fZ$ or $Y=Z$. A nonempty assignment $Y\in \mathcal{A}^f$ is called \textbf{acceptable} for firm $f$ if $Y\succ_f\emptyset$. Let $\overline{\mathcal{A}}^f\equiv\{Y\in\mathcal{A}^f|Y\succ_f\emptyset \}$ be the collection of firm $f$'s acceptable assignments, and $\overline{\mathcal{A}}^F\equiv\cup_{f\in F}\overline{\mathcal{A}}^f$ the collection of firms' acceptable assignments. Let $\succ_F$ and $\succ_W$ be the preference profile of firms and the preference profile of workers, respectively.  A many-to-one matching market is summarized as a tuple: $\Gamma=(F,W,X,\succ_F,\succ_W)$.

A market is in \textbf{basic setting} if there is at most one contract between each firm $f\in F$ and each worker $w\in W$: $|X_f|\cap|X_w|\leq1$. We often provide examples in basic setting (such as the market (\ref{exam_in1})) in which we abuse the notations as follows: For any firms $f,f'\in F$, any worker $w\in W$, and any two subset of workers $S,S'\subseteq W$, we write $S\succ_fS'$ if firm $f$ prefers employing workers from $S$ to employing workers from $S'$, and write $f\succ_wf'$ if worker $w$ prefers working for firm $f$ to working for firm $f'$.

A subset $M\subseteq X$ of contracts is called a \textbf{matching} if $|M_w|\leq1$ for each $w\in W$, that is, if $M$ includes at most one contract for each worker. Therefore, if a set of contracts $M$ is a matching, then its subset $M_f$ is an assignment of firm $f$ for each $f\in F$, and its subset $M_w$ is an assignment of worker $w$ for each $w\in W$.

\begin{definition}\label{def_stable}
\normalfont
An assignment $Y\in \mathcal{A}^f$ of firm $f$ blocks a matching $M$ if $Y\succ_fM_f$, and $Y_w\succeq_w M_w$ for each $w\in W(Y)$. A matching $M$ is \textbf{stable} if it cannot be blocked.
\end{definition}

Notice that firms' individual rationalities of a stable matching are implied by the above definition: In a matching $M$, if firm $f$ wants to unilaterally drop some contracts, then there is $Y\subset M_f$ that blocks $M$. The set of stable matchings defined by Definition \ref{def_stable} is equivalent to the core defined by weak domination of \cite{RS90} and the set of stable matchings of \cite{HM05}. Our model differs from the latter work in that we assume each firm has a preference order over its assignments, while the Hatfield-Milgrom market does not make this assumption. Notice that we use firms' preference orders in condition (\ref{balanced}).

\begin{example}\label{exam_balanced}
\normalfont
Consider a market with two firms and two workers. Firm $f_1$ would like to provide two salaries $4,5$, and two health plans ${c,d}$, for workers. Health plan $c$ is attractive to firm $f_1$ if firm $f_1$ hires only one worker, but plan $d$ becomes more attractive to firm $f_1$ when firm $f_1$ hires two workers. This preference is an instance of the case provided by \cite{E12} that illustrates how a firm and workers negotiate in two dimensions. There are two contracts, $x_{5c}$ and $x_{5d}$, between firm $f_1$ and worker $w_1$, and two contracts, $y_{4d}$ and $y_{5d}$, between firm $f_1$ and worker $w_2$, where the subscripts refer to the salaries and the health plans provided in the contracts. Firm $f_2$ can sign contract $z_1$ with worker $w_1$ and contract $z_2$ with worker $w_2$. The agents have the preferences
\begin{equation*}
\begin{aligned}
&f_1: \{x_{5d},y_{4d}\}\succ\{x_{5d},y_{5d}\}\succ\{x_{5c}\}\succ\emptyset \quad &w_1: &\; \{x_{5d}\}\succ \{z_1\}\succ\{x_{5c}\}\succ\emptyset\\
&f_2: \{z_1,z_2\}\succ\{z_2\}\succ\emptyset \quad &w_2: &\; \{z_2\}\succ\{y_{5d}\}\succ\{y_{4d}\}\succ\emptyset.
\end{aligned}
\end{equation*}
The matching $\{z_2\}$ is not stable since it is blocked by firm $f_2$'s assignment $\{z_1,z_2\}$. In forming this blocking assignment, firm $f_2$ and worker $w_1$ become strictly better off while worker $w_2$'s assignment is unchanged. This market has three stable matchings: $\{x_{5d},y_{4d}\}$, $\{x_{5d},y_{5d}\}$, and $\{z_1,z_2\}$. We show that this market satisfies our concavity condition in Section \ref{Sec_pi}. Both firms' preferences in this example involve complementarities. For instance, firm $f_1$ will not choose contract $x_{5d}$ when neither $y_{4d}$ nor $y_{5d}$ is available but will choose $x_{5d}$ when $y_{4d}$ or $y_{5d}$ is available. Hence, each of the latter two contracts has a complementary effect to the contract $x_{5d}$ for firm $f_1$.\footnote{The contract $x_{5d}$ also has complementary effects to $y_{4d}$ and $y_{5d}$ for firm $f_1$ since neither of the latter two contracts would be chosen by firm $f_1$ when $x_{5d}$ is not available.}
\end{example}

We can find a stable matching using the DA algorithm or its variants when there is no complementarity in firms' preferences or complementarities are in specific forms (see, for example, \citealp{HK10}, \citealp{HK15}, and Section \ref{Sec_App} of this paper). Stable matchings are also characterized as fixed points of certain operators; see \cite{A00}, \cite{F03}, \cite{EO04,EO06}, and \cite{HM05}, among others. When there is no complementarity in agents' preferences, the operators are monotone and converge to stable matchings.

\subsection{Generalized schedule matching\label{Sec_Sch}}

In a many-to-one matching market, a schedule matching of \cite{AG03} allows each worker to schedule her time among her acceptable firms, and each firm to schedule its time among its acceptable sets of workers. We introduce a generalization of a schedule matching. A scheme $\pi=(\pi_Y)_{Y\in\overline{\mathcal{A}}^F\cup\{N\}}$ is a family of vectors where $\pi_Y\in \mathbb{R}^{N}_{+}$, $\pi_Y^{N(Y)}\gg\mathbf{0}$, and $\pi_Y^{N\setminus N(Y)}=\mathbf{0}$ for each $Y\in\overline{\mathcal{A}}^F$, and $\pi_{N}\gg\mathbf{0}$.\footnote{By $\pi_Y\in \mathbb{R}^{N}_{+}$, we mean that each component of $\pi_Y$ is nonnegative. By $\pi_Y^{N(Y)}\gg\mathbf{0}$, we mean that each $\pi_Y(i)$ with $i\in N(Y)$ is positive.} For each firm $f\in F$, the component $\pi_N(f)$ is firm $f$'s capacity for some resource; for each worker $w\in W$, the component $\pi_N(w)$ is the labor supply of worker $w$. For firm $f$'s acceptable assignment $Y\in\overline{\mathcal{A}}^f$, the component $\pi_Y(f)$ is the resource input intensity of firm $f$ in assignment $Y$; if worker $w\in W(Y)$, the component $\pi_Y(w)$ is the labor intensity of worker $w$ when worker $w$ is working in some firm's assignment $Y$.

\begin{definition}\label{def_schedule}
\normalfont
Given a scheme $\pi$, a \textbf{$\pi$-schedule matching} is a vector $\mathbf{t}\in R_+^{\overline{\mathcal{A}}^F}$ satisfying $\sum_{Y\in \overline{\mathcal{A}}^F}\mathrm{t}(Y)\pi_Y(i)\leq\pi_N(i)$ for each $i\in N$.
\end{definition}
A $\pi$-schedule matching $\mathbf{t}$ assigns each acceptable assignment of some firm $Y\in\overline{\mathcal{A}}^F$ a time share $\mathrm{t}(Y)$. Notice that each firm $f\in F$ can schedule its time only over its acceptable assignments. For each firm $f\in F$, the term $\sum_{Y\in \overline{\mathcal{A}}^F}\mathrm{t}(Y)\pi_Y(f)$ is the firm's total resource input; for each worker $w\in W$, the term $\sum_{Y\in \overline{\mathcal{A}}^F}\mathrm{t}(Y)\pi_Y(w)$ is the worker's total workload. Therefore, the constraints in this definition mean that each firm's total resource input does not exceed its resource capacity, and each worker's total workload does not exceed her labor supply. Agent $i\in N$ is called full matched in a $\pi$-schedule matching $\mathbf{t}$ if $\sum_{Y\in \overline{\mathcal{A}}^F}\mathrm{t}(Y)\pi_Y(i)=\pi_N(i)$ holds.

\begin{example}\label{exam_schedule}
\normalfont
A $\pi$ scheme for the market in Example \ref{exam_balanced} is given below.

\medskip

\begin{tikzpicture}[baseline = (M.west)]
    \matrix(M)
    [
        matrix of math nodes,
        left delimiter = (,
        right delimiter = )
    ]
    {
        \qquad 4\qquad & \qquad 2\qquad & \quad 4\quad & \quad 0\quad & \quad 0\quad \\
        0 & 0 & 0 & 2 & 2 \\
        2 & 1 & 2 & 1 & 0 \\
        2 & 1 & 0 & 3 & 3 \\
            };
     \draw[decorate]
         node (f1) at ($(M-1-1) + (-2.1, 0)$) {$f_1$}
         node (f2) at ($(M-2-1) + (-2.1, 0)$) {$f_2$}
         node (w1) at ($(M-3-1) + (-2.1, 0)$) {$w_1$}
         node (w2) at ($(M-4-1) + (-2.1, 0)$) {$w_2$}
         node (f1) at ($(M-1-1) + (0, 0.8)$) {$\{x_{5d},y_{4d}\}$}
         node (f2) at ($(M-1-2) + (0, 0.8)$) {$\{x_{5d},y_{5d}\}$}
         node (w1) at ($(M-1-3) + (0, 0.8)$) {$\{x_{5c}\}$}
         node (w2) at ($(M-1-4) + (0, 0.8)$) {$\{z_1,z_2\}$}
         node (w2) at ($(M-1-5) + (0, 0.8)$) {$\{z_2\}$};
\end{tikzpicture}
\begin{tikzpicture}[baseline = (M.west)]
    \tikzset{brace/.style = {decorate, decoration = {brace, amplitude = 5pt}, thick}}
        \matrix(M)
    [
        matrix of math nodes,
        left delimiter = (,
        right delimiter = )
    ]
    {
       t_1\\
t_2\\
t_3\\
t_4\\
t_5\\
    };
\end{tikzpicture}$\leq$
\begin{tikzpicture}[baseline = (M.west)]
    \tikzset{brace/.style = {decorate, decoration = {brace, amplitude = 5pt}, thick}}
        \matrix(M)
    [
        matrix of math nodes,
        left delimiter = (,
        right delimiter = )
    ]
    {
5\\
3\\
2\\
3\\
    };
\end{tikzpicture}

Vectors of $(\pi_Y)_{Y\in\overline{\mathcal{A}}^F}$ are given by the columns of the matrix at left-hand side, the vector at right-hand side is $\pi_N$. The above system of inequalities is the constraint in Definition \ref{def_schedule}. For example, firm $f_1$'s resource input intensities for the assignments are given by the first row, thus, the inequality $4t_1+2t_2+4t_3\leq5$ means that firm $f_1$'s total resource input does not exceed its resource capacity. Worker $w_1$'s labor intensities for the assignments are given by the third row, thus, the inequality $2t_1+t_2+2t_3+t_4\leq2$ means that worker $w_1$'s total workload does not exceed her labor supply.
\end{example}

For any set of agents $S\subseteq N$, let $\mathrm{ind}(S)\in\{0,1\}^N$ be the indicator vector of the set $S$.\footnote{The indicator vector $\mathrm{ind}(S)$ is a vector in $\{0,1\}^{N}$ such that $\mathrm{ind}(S)(i)=1$ if $i\in S$, and $\mathrm{ind}(S)(i)=0$ if $i\notin S$.} A $\pi$-schedule matching is called a \textbf{schedule matching} if the $\pi$ scheme is defined by $\pi_Y=\mathrm{ind}(N(Y))$ for each $Y\in\overline{\mathcal{A}}^F$ and $\pi_N=\mathbf{1}$. Notice that the time share for each $Y\in\overline{\mathcal{A}}^F$ is a proportion in $[0,1]$ for a schedule matching, while the time shares in a $\pi$-schedule matching are not necessarily proportions.

The notion of the $\pi$-schedule matching is an analogue to the $\pi$-balanced collection of coalitions proposed by \cite{B70}. There is also a difference. The $\pi$-balanced collection of coalitions specifies different parameters over different coalitions, whereas our $\pi$ scheme specifies different parameters over not only different coalitions but also different sets of contracts signed by agents from the same coalition.

\section{Concave markets\label{Sec_pi}}

Let $\mathcal{A}^F_w\equiv\{Y\in\mathcal{A}^F|Y_w\neq\emptyset\}$ be the collection of firms' assignments in which each assignment involves a contract signed by worker $w$. We call elements of $\mathcal{A}^F_w\cup\{\emptyset\}$ \textbf{situations} of worker $w$. We also call firm $f$'s assignments from $\mathcal{A}^f$ situations of firm $f$. In words, a situation of an agent is a set of contracts in which the agent cooperates with other agents. For any worker $w\in W$ and any two of her situations $Y,Z\in \mathcal{A}^F_w\cup\{\emptyset\}$, we write $Y\rhd_w Z$ if (i) $Y_w\succ_wZ_w$, or (ii) $Y_w=Z_w$, and $Y\succ_fZ$ where $\{f\}=Y_f=Z_f$. The order $\rhd_w$ is modified from worker $w$'s preference $\succ_w$ by imposing the artificial externalities stated in (\ref{balanced}). The term $Y\rhd_w Z$ means that, if we consider worker $w$ to be slightly better off when her employer becomes better off, she prefers her situation $Y$ than situation $Z$. Notice that the order $\rhd_w$ for worker $w$ is complete, transitive, and strict over her situations from $\mathcal{A}^F_w\cup\{\emptyset\}$. For any worker $w\in W$ and any two of her situations $Y,Z\in \mathcal{A}^F_w\cup\{\emptyset\}$, we write $Y\unrhd_wZ$ if $Y\rhd_wZ$ or $Y=Z$. For any matching $M$ and any worker $w\in W$, let $M(w)$ be worker $w$'s situation in $M$. That is,
\begin{equation*}
M(w)\equiv\left\{
\begin{aligned}
\{x\in M|\{x_F\}=F(M_w)\},\qquad &\text{ if } M_w\neq\emptyset,\\
\emptyset,\qquad\qquad &\text{ if } M_w=\emptyset.
\end{aligned}
\right.
\end{equation*}

Given a $\pi$-schedule matching $\mathbf{t}$, let $\widetilde{\mathbf{t}}\in(\overline{\mathcal{A}}^F\cup\{\emptyset\})^N$ specify each agent's worst situation in $\mathbf{t}$, where we assume each worker to be slightly better off when her employer becomes better off. In particular, (i) for each firm $f\in F$,
\begin{equation*}
\widetilde{\mathrm{t}}(f)\equiv\left\{
\begin{aligned}
\min_{\succ_f}\{Y\in\overline{\mathcal{A}}^f|\mathrm{t}(Y)>0\},\qquad \text{ if }\sum_{Y\in \overline{\mathcal{A}}^F}\mathrm{t}(Y)\pi_Y(i)=\pi_N(i),\\
\emptyset,\qquad\qquad \text{ if }\sum_{Y\in \overline{\mathcal{A}}^F}\mathrm{t}(Y)\pi_Y(i)<\pi_N(i).
\end{aligned}
\right.
\end{equation*}

(ii) for each worker $w\in W$,
\begin{equation*}
\widetilde{\mathrm{t}}(w)\equiv\left\{
\begin{aligned}
\min_{\rhd_w}\{Y\in\mathcal{A}^F_w|\mathrm{t}(Y)>0\},\qquad \text{ if }\sum_{Y\in \overline{\mathcal{A}}^F}\mathrm{t}(Y)\pi_Y(w)=\pi_N(w),\\
\emptyset,\qquad\qquad \text{ if }\sum_{Y\in \overline{\mathcal{A}}^F}\mathrm{t}(Y)\pi_Y(w)<\pi_N(w).
\end{aligned}
\right.
\end{equation*}

\begin{definition}
\normalfont
A matching $M$ \textbf{dominates} a $\pi$-schedule matching $\mathbf{t}$ if
\begin{description}
\item[(\romannumeral1)] for each $f\in N$, $M_f\succeq_f\widetilde{\mathrm{t}}(f)$;

\item[(\romannumeral2)] for any $w\in W$, $M(w)\unrhd_w\widetilde{\mathrm{t}}(w)$.
\end{description}
\end{definition}
Condition (i) says that each firm receives an assignment in $M$ it weakly prefers to the worst assignment it receives in the $\pi$-schedule matching. Condition (ii) means that, if we consider each worker to be slightly better off when her employer becomes better off, then each worker is weakly better off in $M$ than her worst situation in the $\pi$-schedule matching.

\begin{definition}
\normalfont
A market is called \textbf{$\pi$-concave} if there is a $\pi$ scheme such that for any $\pi$-schedule matching $\mathbf{t}$, there is a matching that dominates $\mathbf{t}$.
\end{definition}

We call a market \textbf{concave} if for any schedule matching $\mathbf{t}$, there is a matching that dominates $\mathbf{t}$. Since a schedule matching is a $\pi$-schedule matching for a specific $\pi$ scheme, a concave market is $\pi$-concave. The converse is not true (see market (\ref{exam_pibalanced}) below).

\begin{theorem}\label{thm_balanced}
\normalfont
A stable matching exists in a $\pi$-concave market.
\end{theorem}

The market in Example \ref{exam_balanced} is concave. To see this, consider a schedule matching $\mathbf{t}$ of this market, if $\mathbf{t}$ is integral (i.e., each component of $\mathbf{t}$ is 0 or 1), then it corresponds to a matching that dominates itself; if $\mathbf{t}$ is not integral and $\mathrm{t}(\{z_2\})=0$, then the matching $\{x_{5d},y_{5d}\}$ dominates $\mathbf{t}$;\footnote{Note that in this case we know that $\mathrm{t}(\{x_{5d},y_{4d}\})<1$, and thus, firm $f_1$'s worst assignment in $\mathbf{t}$ is not $\{x_{5d},y_{4d}\}$. We also know that $\mathrm{t}(\{z_1,z_2\})<1$, and thus, firm $f_2$ is not full matched and $w_2$'s worst assignment in $\mathbf{t}$ is not $\{z_2\}$.} if $\mathbf{t}$ is not integral and $\mathrm{t}(\{z_2\})>0$, then the matching $\{x_{5c},z_2\}$ dominates $\mathbf{t}$.\footnote{Note that $\mathrm{t}(\{z_2\})>0$ implies $\mathrm{t}(\{x_{5d},y_{4d}\})+\mathrm{t}(\{x_{5d},y_{5d}\})<1$.} This market is also $\pi$-concave with respect to the $\pi$ scheme presented in Example \ref{exam_schedule}, which we explain in Section \ref{Sec_Alg}. The next is a market in basic setting that is not concave but is $\pi$-concave.

\begin{equation}\label{exam_pibalanced}
\begin{aligned}
&f_1: \{w_1,w_2\}\succ\emptyset \qquad\qquad\qquad\qquad &w_1: &\quad f_1\succ f_2\\
&f_2: \{w_1\}\succ\{w_2\}\succ\emptyset \qquad\qquad\qquad\qquad &w_2: &\quad f_1\succ f_2
\end{aligned}
\end{equation}
This market is not concave since no matching dominates the schedule matching (\ref{intro_schedule}).\footnote{If a matching dominates (\ref{intro_schedule}), then both workers should be employed. However, the matching $(f_1,w_1,w_2)$ does not dominate (\ref{intro_schedule}).} However, the market is $\pi$-concave with respect to the $\pi$ scheme defined by $\pi_Y=\mathrm{ind}(N(Y))$ for all $Y\in\overline{\mathcal{A}}^F$, and $\pi_N=(1,3,1,1)$ in which the first to fourth components correspond to $f_1,f_2,w_1$, and $w_2$, respectively.\footnote{In any $\pi$-schedule matching $\mathbf{t}$ under this $\pi$ scheme, firm $f_2$ must not be full matched. Then, firm $f_1$ hiring both workers is always a matching that dominates $\mathbf{t}$.} In fact, matching both workers to firm $f_1$ is a stable matching.

We can prove Theorem \ref{thm_balanced} by applying Scarf's lemma (Theorem 2 of \citealp{S67}) to the market with the artificial externalities of \cite{K10}. Scarf's constructive proof of Scarf's lemma also provides a constructive proof for our existence theorem, which we present in Section \ref{Sec_Alg}. In particular, we introduce the stability for $\pi$-schedule matching in Section \ref{Sec_stableSch}. In Section \ref{Sec_Alg}, we use Scarf's algorithm to find a stable $\pi$-schedule matching, then the matching that dominates this $\pi$-schedule matching is stable.

\section{Application\label{Sec_App}}

In this section, we provide a class of concave markets in basic setting. We assume the set of workers $W$ is partitioned into a set of \textbf{leaders} $L$ and a set of \textbf{followers} $O$. Each follower $o\in O$ follows a leader $o_L\in L$. A set of workers is called a \textbf{team} if it contains exactly one leader $l\in L$ and includes only followers who follow this leader $l$. That is, a set of workers $S=\{l\}\cup S'$ with $l\in L$ and $S'\subseteq O$ is called a team if $o_L=l$ for each follower $o\in S'$. Let $\mathcal{S}^T$ be the collection of all teams.

We study the market in which each firm has a unit demand over the teams, that is, $S\succ_f\emptyset$ implies $S\in\mathcal{S}^T$.

\begin{example}\label{exam_app}
\normalfont
Suppose there are two leaders $l_1, l_2$, and three followers $o_1, o_2, o_3$. Follower $o_1$ follows leader $l_1$. Both followers $o_2$ and $o_3$ follow leader $l_2$. There are two firms with the following preferences.
\begin{equation}\label{pre_app}
\begin{aligned}
&f_1: \{l_2,o_2,o_3\}\succ \{l_1,o_1\}\succ \{l_2\}\succ \emptyset\\
&f_2: \{l_2,o_2\}\succ\{l_1,o_1\}\succ \{l_2,o_3\}\succ \emptyset
\end{aligned}
\end{equation}
Notice that each acceptable set of each firm is a team that contains a leader and followers (possibly without any follower) who follow this leader. Hence, each firm has a unit demand over the teams. 
\end{example}

\begin{theorem}\label{thm_app}
\normalfont
The market of matching firms with leaders and followers is concave if firms have unit-demand preferences over the teams.
\end{theorem}

\begin{proof}
Let $\mathcal{Q}_f\equiv\{S\in \mathcal{S}^T|S\succ_f\emptyset$, and $f$ is acceptable to $w$ for each $w\in S\}$ be the collection of firm $f$'s possible sets of employees.
We call a set $\{f\}\cup S$ with $f\in F$ and $S\in\mathcal{Q}_f$ an essential coalition. Let $\mathcal{E}$ be the collection of all essential coalitions. Let $H$ be an $|N|\times(|N|+|\mathcal{E}|)$ matrix whose columns are indicator vectors of elements from $N\cup\mathcal{E}$. A schedule matching in a market in basic setting can be represented by a vector $\mathbf{t}\in \mathbb{R}^{N\cup\mathcal{E}}_+$ with $H\mathbf{t}=\mathbf{1}$.\footnote{Note that the component $\mathrm{t}(j)$ with $j\in N$ is the vacant time share of agent $j$ in the schedule matching $\mathbf{t}$.}

Given a schedule matching $\mathbf{t}^*$, define $J^*\equiv\{j\in(N\cup\mathcal{E})|t^*(j)>0\}$, and define $\mathbf{u}^*\equiv(\mathrm{t}^*(j))_{j\in J^*}$. Let $H'$ be the submatrix of $H$ restricted to row indexes from $F\cup L$ and column indexes from $J^*$. We know that $\mathbf{u}^*$ is in the polytope $P\equiv\{\mathbf{u}|\mathbf{u}\geq\mathbf{0}, H'\mathbf{u}=\mathbf{1}\}$. Matrix $H'$ is the incidence matrix of a bipartite graph that has no odd-length cycle, and thus, matrix $H'$ balanced.\footnote{A $\{0,1\}$-matrix is balanced if it has no square submatrix of odd order with exactly two 1s in each row and column; see Chapter 21.5 of \cite{S86}. Let $G'$ be a bipartite graph whose vertices are divided into $F$ and $L$. An edge $fl$ is in $G'$ if there is a column of $H'$ containing a component 1 for $f$ and a component 1 for $l$. Then, matrix $H'$ is the incidence matrix of $G'$.} Then, since the polytope $P$ is nonempty, we know that there is at least an integral vertex $\mathbf{u}^{**}\in\{0,1\}^{J^*}$ on the polytope $P$.\footnote{See Lemma 2.1 of \cite{FHO74}.} Vertex $\mathbf{u}^{**}$ induces the following full-time matching $M$: if $\mathrm{u}^{**}(j)=1$ where $j=(\{f\}\cup S)\in\mathcal{E}$, then match $f$ with $S$; other firms and workers stay unmatched. Notice that we have not matched a leader with more than one firm, then since each follower follows one leader, we have not matched a follower with more than one firm either. Hence, $M$ is a matching. The matching $M$ dominates the schedule matching $\mathbf{t}^*$ since (i) each essential coalition carried out by the matching $M$ has a positive time share in $\mathbf{t}^*$, and (ii) each unmatched agent in $M$ is not full matched in $\mathbf{t}^*$. The statement (ii) is due to the following reasons: If agent $i\in F\cup L$ is not matched in the matching $M$, column $i$ is in the matrix $H'$, and thus, agent $i$ is not full matched in $\mathbf{t}^*$. Suppose a follower $o\in O$ not matched in the matching $M$ is full matched in $\mathbf{t}^*$. Then her leader $o_L$ is also full matched in $\mathbf{t}^*$. We know that column $o_L$ is not in the matrix $H'$, and the column indexes of matrix $H'$ that contains $o_L$ must contain $o$. Thus, the follower $o$ and her leader $o_L$ must be matched with the same firm in the matching $M$. A contradiction. Therefore, if follower $o\in O$ is not matched in the matching $M$, she is not full matched in the schedule matching $\mathbf{t}^*$.
\end{proof}

The market of this section can be viewed as a generalization of a marriage market: A stable matching is essentially a ``stable marriage'' between the firms and the leaders. We can find a stable matching using the following variant DA algorithm. For each firm $f\in F$ and each leader $l\in L$, define
\begin{equation*}
\mathrm{U}(f,l)\equiv\max_{\succ_f}\{\emptyset\}\cup\{S\in\mathcal{S}^T|l\in S, \text{ and } f \text{ is acceptable to all workers from }S\}.
\end{equation*}
to be firm $f$'s favorite set of possible employees that includes leader $l$.
For each firm $f\in F$ and any two leaders $l,l'\in L$, we write $l\geq_fl'$ if $\mathrm{U}(f,l)\succeq_f\mathrm{U}(f,l')$. The variant DA algorithm proceeds as follows.
\medskip
\begin{description}
\item[Step $1$] Each leader $l\in L$ proposes to her most preferred firm $f$. Each firm tentatively accept its most preferred worker according to the order $\geq_f$ among those applicants satisfying $\mathrm{U}(f,l)\neq\emptyset$. Each leader not accepted in this round is rejected by the firm.

\item[Step $k,k\geq2$] Each leader rejected in the previous round proposes to her next best firm, so long as there remains an acceptable firm to whom she has not yet proposed. Each firm that faces new applicants tentatively accepts its most preferred leader according to the order $\geq_f$ among both new applicants satisfying $\mathrm{U}(f,l)\neq\emptyset$ and previously accepted leader. Each leader not accepted in this round is rejected by the firm. 

The algorithm stops when there are no rejections.
\end{description}
\medskip
The algorithm matches each firm $f$ with $\mathrm{U}(f,l)$ if firm $f$ accepts leader $l$ at the end of the algorithm. Other firms and workers stay unmatched. We leave it for the reader to prove that the algorithm produces a stable matching.

Using an argument similar as the one used in the proof of Theorem \ref{thm_app}, we can prove that the discrete matching market of \cite{H22} with a balanced firm-worker hypergraph is concave.\footnote{We can use Lemma 2.1 of \cite{FHO74} and a similar argument as that in the proof of Theorem \ref{thm_app} to prove the following statement: If a market has a balanced firm-worker hypergraph, then for any schedule matching, there is a full-time matching in which each essential coalition carried out by this full-time matching has a positive time share in the schedule matching, and each agent unmatched in the full-time matching is not full matched in the schedule matching.} 

The market of this section subsumes the market in Section 4 of \cite{H23}\footnote{In the market in Section 4 of \cite{H23}, we can let the set of leaders be the set that contains exactly one worker from each technology incident to the root and let the followers be the rest workers. Then this market is an instance of our market.} and is independent of the conditions of \cite{H22,H23}. The market (\ref{pre_app}) has neither a totally unimodular demand type nor a balanced firm-worker hypergraph according to the nontrivial odd-length cycle $(f_2$, $\{f_2,l_2,o_2\}$, $o_2$, $\{o_2,l_2,f_1,o_3\}$, $o_3$, $\{o_3,l_2,f_2\}$, $f_2)$ in its firm-worker hypergraph.\footnote{This nontrivial odd-length cycle indicates that the market (\ref{pre_app}) fails the condition of Proposition 1 of \cite{H22}, then Proposition 2 therein implies that the market (\ref{pre_app}) does not have a totally unimodular demand type either.}

It is natural to generalize the market by allowing firms' substitutable preferences over the teams. In particular, the teams are substitutable for a firm if the following statement holds: if the firm hires a team (possibly with some other teams) from a set of available workers, the team will still be hired when any worker outside the team becomes unavailable. We also leave it for the reader to prove the following result: If each firm has a substitutable preference over the teams, a stable matching exists and can be found by another variant DA algorithm.

If there is no follower, the preferences defined above reduce to the substitutable preferences.\footnote{A firm has a substitutable preference if there is no complementarity in its preference. More specifically, a firm has a substitutable preference if any worker chosen by the firm from a set of available workers will still be chosen as the available set shrinks. See Chapter 6 of \cite{RS90}.} It is unknown whether the market with substitutable preferences and our extension are $\pi$-concave. A related result concerns the matching market in which each agent has a transferable utility. The characterization result of \cite{M98} implies that the Bondareva-Shapley balancedness condition holds if each firm has a gross-substitute valuation.

\section{Stable $\pi$-schedule matching}\label{Sec_stableSch}

In this section, we describe agents' preferences in schedule matching under a $\pi$ scheme and define stability for $\pi$-schedule matching. We assume each worker wants to work for her favorite assignment as much as possible, and then work for her second favorite assignment as much as possible, and so on. We assume that a firm's different assignments bring to the firm different profits per unit resource, which decrease along with the firm's preference order $\succ_f$. This assumption induces a best arrange of time shares for each firm over its acceptable assignments  when the firm is offered labor supplies from workers. In particular, for each firm $f\in F$, suppose its acceptable assignment $Y\in \overline{\mathcal{A}}^f$ brings $\mathrm{g}_f(Y)$ per unit resource to firm $f$. We assume $\mathrm{g}_f(Y)>\mathrm{g}_f(Z)>0$ if $Y\succ_fZ\succ_f\emptyset$. Given a $\pi$ scheme and a vector $\mathbf{q}\in R^{W}_{+}$ of labor supplies from workers, each firm $f\in F$ solves the following problem, where the decision variable is firm $f$'s time shares $\mathbf{s}\in R^{\overline{\mathcal{A}}^f}_{+}$ over its acceptable assignments from $\overline{\mathcal{A}}^f$.

\begin{equation}\label{sch_program}
  \begin{aligned}
  \text{maximize } \qquad & \sum_{Y\in\overline{\mathcal{A}}^f}\mathrm{g}_f(Y) \pi_Y(f)\mathrm{s}(Y) \\
  \text{subject to } \qquad & \sum_{Y\in\overline{\mathcal{A}}^f}\pi_Y(f)\mathrm{s}(Y)\leq\pi_N(f) & \quad \\
  \quad \qquad & \sum_{Y\in\overline{\mathcal{A}}^f}\pi_Y(w)\mathrm{s}(Y)\leq \mathrm{q}(w), & \text{ for each } w\in W \\
  \quad \qquad & \mathrm{s}(Y)\geq0, & \text{ for each } Y\in \overline{\mathcal{A}}^f
  \end{aligned}
\end{equation}
When firm $f$ is offered a vector $\mathbf{q}$ of labor supplies from workers, the firm's best arrange of time shares over its acceptable assignments is the solution to the problem (\ref{sch_program}).

In the following discussion, when we talk about the preference of some firm $f$ or some worker $w$ over its assignments, we mean its preference $\succ_f$ or $\succ_w$ over full-time assignments. Intuitively, a $\pi$-schedule matching $\mathbf{t}$ is blocked when some firm $f$ draw workers to carry out an acceptable assignment $Z\in\overline{\mathcal{A}}^f$ of firm $f$ for an additional amount of time share. Some time share of other assignment in $\mathbf{t}$ may be dismissed at the same time to meet firm $f$'s resource constraint or some worker's labor constraint. The firm is willing to participate in a blocking assignment as long as the net effect in profit is positive. We know that (i) if firm $f$ is solving the problem (\ref{sch_program}), it would like to participate in a blocking assignment $Z\in\overline{\mathcal{A}}^f$ when firm $f$ prefers $Z$ to its worst assignment in $\mathbf{t}$: $Z\succ_f\widetilde{\mathrm{t}}(f)$. This is because in this case $Z$ brings larger profit per unit resource than $\widetilde{\mathrm{t}}(f)$, and thus, firm $f$ would like to shift its resource from $\widetilde{\mathrm{t}}(f)$ into $Z$. (ii) The workers drawn by firm $f$ are willing to participate in the blocking assignment $Z$ if each worker $w$ from $W(Z)$ prefers $Z_w$ to her worst assignment in $\mathbf{t}$. (iii) The last question concerns whether worker $w$, whose worst assignment in $\mathbf{t}$ is $Z_w$, can be involved in the blocking assignment $Z$. The answer is positive if, in the schedule matching $\mathbf{t}$, worker $w$'s worst assignment $Z_w$ is from an assignment of firm $f$ that is less preferred than $Z$ by firm $f$. For example, consider a $\pi$-schedule matching $\mathbf{t}=(\frac{1}{3},\frac{1}{3},\frac{1}{3},\frac{1}{3},\frac{1}{3})$ in Example \ref{exam_schedule}. Worker $w_2$'s worst assignment in $\mathbf{t}$ is $\{z_2\}$. Since firm $f_1$ prefers $\{z_1,z_2\}$ to $\{z_2\}$, firm $f_1$ would like to expand its assignment $\{z_2\}$ to $\{z_1,z_2\}$. Since worker $w_1$ is willing to participate in the blocking assignment $\{z_1,z_2\}$, and worker $w_2$'s welfare is unchanged in the expansion from $\{z_2\}$ to $\{z_1,z_2\}$, the blocking assignment $\{z_1,z_2\}$ can be implemented. We can thus summarize (ii) and (iii) as follows: A worker $w$ can be involved in a blocking assignment $Z\in\mathcal{A}^F_w$ if $Z\rhd_w\widetilde{\mathrm{t}}(w)$. Therefore, we have the following definition of stability for $\pi$-schedule matching.

\begin{definition}\label{def_stablesch}
\normalfont
An assignment $Z\in\overline{\mathcal{A}}^f$ of some firm $f\in F$ blocks a $\pi$-schedule matching $\mathbf{t}$ if $Z\succ_f\widetilde{\mathrm{t}}(f)$ and $Z\rhd_w\widetilde{\mathrm{t}}(w)$ for all $w\in W(Z)$.
A $\pi$-schedule matching is stable if it cannot be blocked.
\end{definition}

In other words, if we consider each worker to be slightly better off when her employer becomes better off, some firm's acceptable assignment $Z\in \overline{\mathcal{A}}^F$ blocks a $\pi$-schedule matching $\mathbf{t}$ if each agent $i\in N(Z)$ is better off in $Z$ than in its worst situation in $\mathbf{t}$.

In the case of schedule matching, a firm's optimal arrange of time shares is the firm's best choice from workers' supplies of working time. In this case, Definition \ref{def_stablesch} is equivalent to the definition of stability in \cite{CKK19} using firms' choice functions. Since in schedule matching firms' choice functions derived from the problem (\ref{sch_program}) are continuous, the existence of a stable schedule matching follows from the existence theorem of \cite{CKK19}.

The existence of a stable $\pi$-schedule matching follows from Scarf's lemma. We show how to find a stable $\pi$-schedule matching using Scarf's algorithm in Section \ref{Sec_Alg}. Then, if the market is $\pi$-concave, the matching that dominates this stable $\pi$-schedule matching is stable. To see the last statement, suppose, on the contrary, a matching $M$ dominates a stable $\pi$-schedule matching $\mathbf{t}$ and is blocked by firm $f$'s assignment $Z\in\mathcal{A}^f$. Since $Z$ blocks $M$, we have $Z\succ_f M_f$ and $Z_w\succeq_wM_w$ for all $w\in W(Z)$, which further imply $Z\rhd_wM(w)$ for all $w\in W(Z)$. Since $M$ dominates $\mathbf{t}$, we have $M_f\succeq_f\widetilde{\mathrm{t}}(f)$ and $M(w)\unrhd_w\widetilde{\mathrm{t}}(w)$ for each $w\in W$. Hence, for firm $f$, $Z\succ_f M_f$ and $M_f\succeq_f\widetilde{\mathrm{t}}(f)$ implies $Z\succ_f\widetilde{\mathrm{t}}(f)$, which further implies $Z\in\overline{\mathcal{A}}^f$; for each worker $w\in W(Z)$, $Z\rhd_wM(w)$ and $M(w)\unrhd_w\widetilde{\mathrm{t}}(w)$ implies $Z\rhd_w\widetilde{\mathrm{t}}(w)$. Therefore, the $\pi$-schedule matching $\mathbf{t}$ is blocked by $Z$. A contradiction.

One can also strengthen condition (\ref{balanced}) as: for any stable schedule matching, there is a matching that dominates this stable schedule matching.\footnote{Or, more generally, for any stable $\pi$-schedule matching, there is a matching that dominates this stable $\pi$-schedule matching.} However, the strengthened condition becomes more difficult to verify. The conditions of \cite{H22,H23} are special cases of this strengthened condition and are easy to verify.\footnote{See Theorem 1 of \cite{H23} and Proposition 1 of \cite{H22}, where the latter implies the former. The condition of Proposition 1 of \cite{H22} guarantees the transformation from an arbitrary stable schedule matching into a stable matching. The produced stable matching has the following property: Any agent matched in a coalition in the stable matching has a positive time share for this coalition in the stable schedule matching, and any agent unmatched in the stable matching is not full matched in the stable schedule matching. Thus, the produced stable matching dominates the original stable schedule matching.}

\section{Scarf's algorithm\label{Sec_Alg}}

This section shows how to use Scarf's algorithm to find a stable $\pi$-schedule matching. We will use the market in Example \ref{exam_balanced} and the $\pi$ scheme for this market in Example \ref{exam_schedule} to illustrate the algorithm. This market is $\pi$-concave with respect to the $\pi$ scheme. To see this, notice that the constraint for worker $w_1$'s workload implies $2b_1+b_2+2b_3\leq2$, which further implies that firm $f_1$ must not be full matched in a $\pi$-schedule matching. Similarly, the constraint for worker $w_2$'s workload implies $3b_4+3b_5\leq3$, which further implies that firm $f_2$ must not be full matched in a $\pi$-schedule matching. Hence, for any $\pi$-schedule matching $\mathbf{t}$, (i) if $2\mathrm{t}(\{x_{5d},y_{4d}\})+\mathrm{t}(\{x_{5d},y_{5d}\})<2$, worker $w_1$'s worst situation in $\mathbf{t}$ is not $\{x_{5d},y_{4d}\}$ or $\{x_{5d},y_{5d}\}$, then the matching $\{z_1,z_2\}$ dominates $\mathbf{t}$; (ii) if $2\mathrm{t}(\{x_{5d},y_{4d}\})+\mathrm{t}(\{x_{5d},y_{5d}\})=2$, then the matching $\{x_{5d},y_{5d}\}$ dominates $\mathbf{t}$.

We construct matrix $A$ and matrix $C$ for Theorem 2 of \cite{S67} as follows. Let $A$ be an $n\times(n+|\overline{\mathcal{A}}^F|)$ matrix in which each row represents an agent $i\in N$. The first $n$ columns of $A$ correspond to the agents from $N$, and the rest columns correspond to firms' acceptable assignments from $\overline{\mathcal{A}}^F$. Let $a_{ij}$ with $i\in N$ and $j\in N\cup\overline{\mathcal{A}}^F$ denote the elements of $A$. Let column $j$ of matrix $A$ be $\mathrm{ind}(\{j\})$ for each $j\in N$, and let column $Y$ of matrix $A$ be $\pi_Y$ for each $Y\in\overline{\mathcal{A}}^F$. For instance, the matrix $A$ for the market of Example \ref{exam_balanced} with respect to the $\pi$ scheme of Example \ref{exam_schedule} is
\begin{center}
\begin{tabular}
[c]{c|ccccccccc}
& \quad$f_1$\; & \;$f_2$\; & \;$w_1$\; & \;$w_2$ & $\{x_{5d},y_{4d}\}$ & $\{x_{5d},y_{5d}\}$ & $\{x_{5c}\}$ & $\{z_1,z_2\}$ & $\{z_2\}$\\\hline
$f_1$ & 1 & 0 & 0 & 0 & 4 & 2 & 4 & 0 & 0\\
$f_2$ & 0 & 1 & 0 & 0 & 0 & 0 & 0 & 2 & 2\\
$w_1$ & 0 & 0 & 1 & 0 & 2 & 1 & 2 & 1 & 0\\
$w_2$ & 0 & 0 & 0 & 1 & 2 & 1 & 0 & 3 & 3
\end{tabular}
\end{center}

Matrix $C$ has the same dimension as matrix $A$. Let $c_{ij}$ with $i\in N$ and $j\in N\cup\overline{\mathcal{A}}^F$ denote the elements of matrix $C$. Each element $c_{ii}$ with $i\in N$ represents the ``utility" obtained by agent $i$ when $i$ is assigned with $\emptyset$. Each element $c_{ij}$ with $j\in \overline{\mathcal{A}}^F$ and $i\in N(j)$ represents the ``utility" obtained by agent $i$ in the assignment $j$ in the market with artificial externalities. In particular,

(i) let $c_{ii}=0$ for all $i\in N$; other elements of matrix $C$ are larger than 0;

(ii) for any firm $f\in F$ and any two of its acceptable assignments $Y,Z\in\overline{\mathcal{A}}^f$, $c_{fY}>c_{fZ}$ if $Y\succ_fZ$;

(iii) for any worker $w\in W$ and any $Y,Z\in\overline{\mathcal{A}}^F$ with $w\in W(Y)$ and $w\in W(Z)$, $c_{wY}>c_{wZ}$ if (a) $Y_w\succ_wZ_w$, or (b) $Y_w=Z_w$ and $Y\succ_fZ$ where $\{f\}=F(Y)=F(Z)$;

(iv) let all $c_{ij}$ with $i,j\in N$ and $i\neq j$ and all $c_{kY}$ with $Y\in \overline{\mathcal{A}}^F$ and $k\notin N(Y)$ be numbers larger than the largest number among those used in (ii) and (iii);

(v) in (ii), (iii), and (iv), we choose numbers such that $c_{ij}\neq c_{ik}$ for all $i\in N$, $j,k\in N\cup\overline{\mathcal{A}}^F$ with $j\neq k$ (i.e., the numbers in the same row are different), and $c_{ij}>c_{iY}$ for all $Y\in \overline{\mathcal{A}}^F$ and $i,j\in N$ with $i\neq j$ (i.e., the nondiagonal elements of the first $n$ columns are larger than the elements of the columns from $\overline{\mathcal{A}}^F$).

For instance, we construct matrix $C$ for the illustrative example below.
\begin{center}
\begin{tabular}
[c]{c|ccccccccc}
& $\;f_1\;$ & $\;f_2\;$ & $\;w_1\;$ & $\;w_2\;$ & $\{x_{5d},y_{4d}\}$ & $\{x_{5d},y_{5d}\}$ & $\{x_{5c}\}$ & $\{z_1,z_2\}$ & $\{z_2\}$\\\hline
$f_1$ & 0 & $L_2$ & $L_3$ & $L_4$ & 9 & 8 & 7 & $L_8$ & $L_9$\\
$f_2$ & $L_1$ & 0 & $L_3$ & $L_4$ & $L_5$ & $L_6$ & $L_7$ & $8$ & $6$\\
$w_1$ & $L_1$ & $L_2$ & 0 & $L_4$ & $6+\epsilon$ & 6 & 4 & 5 & $L_9$\\
$w_2$ & $L_1$ & $L_2$ & $L_3$ & 0 & 3 & 5 & $L_7$ & $8+\epsilon$ & 8
\end{tabular}
\end{center}

The $L$'s are arbitrary other than satisfying $L_1>L_2>\ldots>L_8>L_9>9$.
Notice that firm $f_1$ prefers $\{x_{5d},y_{4d}\}$ to $\{x_{5d},y_{5d}\}$, then the $\epsilon$ added in $c_{w_1\{x_{5d},y_{4d}\}}$ (as compared to $c_{w_1\{x_{5d},y_{5d}\}}$) refers to the artificial externality of worker $w_1$ being slightly better off when her employer $f_1$ becomes better off. We add an $\epsilon$ in $c_{w_2\{z_1,z_2\}}$ (as compared to $c_{w_2\{z_2\}}$) for a similar reason.

\begin{definition}\label{def_pivot}
\normalfont
\begin{description}
\item[(\romannumeral1)] Given a feasible basis of $\{\mathbf{b}|\mathbf{b}\geq\mathbf{0}, A\mathbf{b}=\pi_N\}$ and a column $j$ of matrix $A$ outside this basis, a \textbf{pivot step} produces a new feasible basis of $\{\mathbf{b}|\mathbf{b}\geq\mathbf{0}, A\mathbf{b}=\pi_N\}$ by replacing a column of the basis with the column $j$.
\item[(\romannumeral2)] An \textbf{ordinal basis} of the matrix $C$ consists of a set of $n$ columns $k_1,k_2,\ldots,k_n\in N\cup\overline{\mathcal{A}}^F$ such that if $u_i=\min(c_{ik_1},c_{ik_2},\ldots,c_{ik_n})$, then for every column $s\in N\cup\overline{\mathcal{A}}^F$, there is at least one $i$ with $u_i\geq c_{is}$.
\item[(\romannumeral3)] Given an ordinal basis of matrix $C$ and a column $j$ from this basis, an \textbf{ordinal pivot step} produces a new ordinal basis by replacing the column $j$ by a column outside the basis as follows: When the column $j$ has been removed, in the $n\times(n-1)$ matrix of remaining columns precisely one column will contain two row minimizers, one of which is new and the other a row minimizer for the original basis. Let the row associated with the latter have an index $i^*$. Examing all columns in matrix $C$ for which
    \begin{equation}\label{ordinal}
    c_{ik}>\min\{c_{is}|s\text{ remains in the basis}\}
    \end{equation}
    holds for all $i$ not equal to $i^*$. Of these columns, select the one which maximizes $c_{i^*k}$. Introduces this column into the $C$ basis.
\end{description}
\end{definition}

We assume the polytope $\{\mathbf{b}|\mathbf{b}\geq\mathbf{0}, A\mathbf{b}=\pi_N\}$ is nondegenrate, which means that for any basic feasible solution to the problem, all $n$ components associated with the corresponding basis are positive. We can deal with degenerate cases by perturbing the polytope.\footnote{See, for example, Chapter 3 of \cite{V14}.} According to Lemma 1 of \cite{S67}, under the nondegenracy assumption, there is a unique column to remove in a pivot step and therefore a unique way to implement a pivot step.

An ordinal basis for matrix $C$ has the following property: Each column contains exactly one row minimizer of the basis.\footnote{Otherwise, there is a column of the basis with no row minimizer, then each component of this column is greater than the corresponding row minimizer of the basis.} Consider a collection of columns of matrix $C$ with the largest $n$ components for some row $i\in N$. Such a collection is an ordinal basis for matrix $C$. For example, in the above matrix $C$, the columns $f_2,w_1,w_2,\{z_1,z_2\}$ have the largest four components for row $f_1$. These columns form an ordinal basis since each column outside the basis has a row-$f_1$-component of less than $L_8$. An ordinal basis of this kind consists of $n-1$ columns from $N$ and one column from $\overline{\mathcal{A}}^F$. Moreover, there is no column outside this basis that forms an ordinal basis with these $n-1$ columns from $N$.\footnote{For example, the columns $f_2,w_1$, and $w_2$ and a column other than $\{z_1,z_2\}$ does not form an ordinal basis since each of its row minimizers is smaller than the corresponding component of the column $\{z_1,z_2\}$.}

An ordinal pivot step can be implemented if, after removing the column $j$, the remaining $n-1$ columns are not all from $N$, since in this case column $i^*$ (the index $i^*$ is identified in Definition \ref{def_pivot} (iii)) is always a candidate outside the basis that satisfies (\ref{ordinal}).\footnote{If the remaining $n-1$ columns are not all from $N$, the column that has two row minimizers is not from $N$, and we then know that column $i^*$ is neither in the remaining $n-1$ columns nor the column removed.} Otherwise, if all the remaining $n-1$ columns are from $N$, the ordinal pivot step cannot be implemented since there is no qualified column to brought into the basis. The reader can then check that if an ordinal pivot step can be implemented, it indeed transforms an ordinal basis for matrix $C$ into another ordinal basis. We are now ready to introduce Scarf's algorithm for finding a stable $\pi$-schedule matching.

\bigskip
Step 0. We begin with the feasible basis for $A$ that consists of the first $n$ columns of $A$, and an ordinal basis for matrix $C$ that consists of columns with the largest $n$ components for some row. Let $A^0=N$ denote the set of column indexes of the $A$ basis, $C^0$ the set of column indexes of the $C$ basis, and $j^C_0$ the column index defined by $\{j^C_0\}\equiv C^0\setminus A^0$.

Step $s, s\geq1$. Implement a pivot step on $A^{s-1}$ by bringing column $j^C_{s-1}$ into the $A$ basis. Suppose column $j^A_{s-1}$ is removed from the $A$ basis in the pivot step, and let $A^s$ denote the set of column indexes of the new $A$ basis. If $A^s=C^{s-1}$, the algorithm terminates. Otherwise, implement an ordinal pivot step on $C^{s-1}$ by replacing column $j^A_{s-1}$ with a column outside $C^{s-1}$. Suppose column $j^C_s$ is brought into the $C$ basis in the ordinal pivot step, and let $C^s$ denote the set of column indexes of the new $C$ basis. If $C^s=A^s$, the algorithm terminates. Otherwise, proceed to the next step.

\bigskip

\begin{example}\label{exam_alg}
\normalfont
We show how to find a stable $\pi$-schedule matching in the illustrative example using the above algorithm.

Step 0. We choose the ordinal basis of $C$ with the largest four components for the row $f_1$. We have $A^0=\{f_1,f_2,w_1,w_2\}$, $C^0=\{\{z_1,z_2\},f_2,w_1,w_2\}$, and $j^C_0=\{z_1,z_2\}$.

Step 1. The pivot step on $A^0$ brings column $\{z_1,z_2\}$ and removes column $w_2$. We have $A^1=\{f_1,f_2,w_1,\{z_1,z_2\}\}\neq C^0$. The ordinal pivot step then removes column $w_2$ from the $C$ basis. In the remaining three columns, column $\{z_1,z_2\}$ has two row minimizers, with the old one in row $f_1$. We thus examine all columns $k$ with $c_{f_2k}>0, c_{w_1k}>0$, and $c_{w_2k}>8+\epsilon$ and select the one which maximizes $c_{f_1k}$. This is column $\{x_{5c}\}$. We have $C^1=\{\{x_{5c}\},f_2,w_1,\{z_1,z_2\}\}\neq A^1$.

Step 2. The pivot step on $A^1$ brings column $\{x_{5c}\}$ and removes column $w_1$. We have $A^2=\{f_1,f_2,\{x_{5c}\},\{z_1,z_2\}\}\neq C^1$. The ordinal pivot step then removes column $w_1$ from the $C$ basis. In the remaining three columns, column $\{x_{5c}\}$ has two row minimizers, with the old one in row $f_1$. We thus examine all columns $k$ with $c_{f_2k}>0, c_{w_1k}>4$, and $c_{w_2k}>8+\epsilon$ and select the one which maximizes $c_{f_1k}$. This is column $f_1$. We have $C^2=\{f_1,f_2,\{x_{5c}\},\{z_1,z_2\}\}=A^2$. The algorithm terminates.
\end{example}

Notice that $A^0$ and $C^0$ differ in only one element. For example, if we choose $C^0$ to be the collection of columns with the largest $n$ components for row $f_1$, then all columns from $A^0$ are in $C^0$ except for column $f_1$. The algorithm maintains this relation until column $f_1$ is removed from the $A$ basis at a pivot step or column $f_1$ is brought into the $C$ basis at an ordinal pivot step, and in either case the algorithm terminates.

A critical point is that the algorithm never cycles. To see this, suppose (i) the algorithm produces $(A^0,C^0)$ at some step $s\geq2$, then at the ordinal pivot step for $C^{s-1}$, after removing column $j^A_{s-1}$, all the remaining $n-1$ columns are from $N$,\footnote{This is because the $n-1$ columns of $A^0\cap C^0$ are from $N$.} and thus this ordinal pivot step cannot be implemented; (ii) the algorithm does not cycle until Step $s$ where Step $s$ produces $(A^s,C^s)=(A^k,C^k)$ with $s>k\geq1$, then Lemma 1 and Lemma 2 of \cite{S67} imply $(A^{s-1},C^{s-1})=(A^{k-1},C^{k-1})$.\footnote{Since both $C^{s-1}$ and $C^{k-1}$ can be obtained from $C^k$ by replacing the column from $C^k\setminus A^k$ with a column outside $C^k$ (note that and $A^k\neq C^k$), Lemma 2 of \cite{S67} implies $C^{s-1}=C^{k-1}$. Then we know that both $A^{s-1}$ and $A^{k-1}$ can be obtained from $A^k$ by replacing a column with the column from $C^{k-1}\setminus C^k$, thus, Lemma 1 of \cite{S67} implies $A^{s-1}=A^{k-1}$.} This contradicts the assumption that the algorithm does not cycle before Step $s$. Since the problem is finite, and the algorithm never cycles, the algorithm must terminate.

\begin{theorem}
\normalfont
Scarf's algorithm produces a stable $\pi$-schedule matching.
\end{theorem}

The algorithm terminates at a set of column indexes that corresponds to both a feasible basis for $\{\mathbf{b}|\mathbf{b}\geq\mathbf{0}, A\mathbf{b}=\pi_N\}$ and an ordinal basis for matrix $C$. Let $\mathbf{b}$ be the basic feasible solution for this feasible basis, then the vector $\mathbf{t}=(\mathrm{b}(Y))_{Y\in\overline{\mathcal{A}}^F}$ is a $\pi$-schedule matching.  Recall that each entity of matrix $C$ is the ``utility'' obtained by the row agent in the column situation in the market with artificial externalities. The definition of an ordinal basis for matrix $C$ implies that $\mathbf{t}$ is a stable $\pi$-schedule matching. In Example \ref{exam_alg}, the basic feasible solution of $\{\mathbf{b}|\mathbf{b}\geq\mathbf{0}, A\mathbf{b}=\pi_N\}$ for the feasible basis $A_2$ is $\mathbf{b}=(3,1,0,0,0,0,\frac{1}{2},1,0)$. Thus, we obtain a stable $\pi$-schedule matching $\mathbf{t}=(0,0,\frac{1}{2},1,0)$, which assigns time share $\frac{1}{2}$ to $\{x_{5c}\}$ and time share $1$ to $\{z_1,z_2\}$.

Recall that the market in the example is $\pi$-concave, we find that the matching $\{z_1,z_2\}$ dominates the above stable $\pi$-schedule matching $\mathbf{t}$. Therefore, the matching $\{z_1,z_2\}$ is stable.

\end{document}